\documentclass[11pt]{amsart}

\usepackage{a4wide}
\usepackage{graphicx}
\usepackage{amscd}
\usepackage{amsmath}
\usepackage{amsfonts}
\usepackage{amssymb}
\usepackage{setspace}
\usepackage{enumerate}         
\usepackage{color}
\usepackage{url}
\usepackage{amsthm}
\usepackage{hyperref}
\usepackage{bm}
\usepackage{xy}
\usepackage{stmaryrd}

\theoremstyle{plain}
\newtheorem{theorem}{Theorem}[section]

\newtheorem{lemma}[theorem]{Lemma}

\newtheorem{proposition}[theorem]{Proposition}

\newtheorem{definition}[theorem]{Definition}

\newtheorem{assumption}[theorem]{Assumption}
\theoremstyle{remark}
\newtheorem{remark}[theorem]{Remark}

\numberwithin{equation}{section}

\newcommand{\ind}{1\!\kern-1pt \mathrm{I}}
\newcommand{\rsto}{]\!\kern-1.8pt ]}
\newcommand{\lsto}{[\!\kern-1.7pt [}

\numberwithin{equation}{section}

\newcommand{\abs}[1]{\left\vert#1\right\vert}

\newcommand{\RR}{\mathbb{R}}

\newcommand{\NN}{\mathbb{N}}

\newcommand{\cA}{\mathcal{A}}

\newcommand{\cC}{\mathcal{C}}
\newcommand{\cD}{\mathcal{D}}

\newcommand{\cF}{\mathcal{F}}
\newcommand{\cM}{\mathcal{M}}

\newcommand{\cZ}{\mathcal{Z}}

\newcommand{\massP}{\mathbf{P}}
\newcommand{\massQ}{\mathbf{Q}}
\newcommand{\massE}{\mathbf{E}}

\begin{document}

\title[Utility maximization under transaction costs and bounded random endowment]{A note on utility maximization\\ with transaction costs and random endowment: num\'eraire-based model and convex duality*\footnote{*The authors have no intention of seeking publishing opportunities for this note, which is organized only for the convenience of communication.}}

\author{Lingqi Gu}
\author{Yiqing Lin}
%
\author{Junjian YANG}
\address{Fakult\"at f\"ur Mathematik\newline \indent Universit\"at Wien\newline \indent Oskar-Morgenstern Platz 1\newline \indent A-1090 Wien, Austria\newline}
\email{lingqi.gu@univie.ac.at; junjian.yang@univie.ac.at}

\begin{abstract}
 In this note, we study the utility maximization problem on the terminal wealth under proportional transaction costs and bounded random endowment. 
 In particular, we restrict ourselves to the num\'eraire-based model and work with utility functions only supporting $\mathbb{R}_+$. 
 Under the assumption of existence of consistent price systems and  natural regularity conditions, standard convex duality results are established. 
 Precisely, we first enlarge the dual domain from the collection of martingale densities associated with consistent price systems to a set of finitely additive measures; 
   then the dual formulation of the utility maximization problem can be regarded as an extension of \cite{CSW01} to the context under proportional transaction costs.
\end{abstract}
\date{\today}

\maketitle

\section{Introduction}
 
\noindent Utility maximization under proportional transaction costs is a classical problem in Mathematical Finance. 
In general, this problem is investigated by two major approaches: dynamic programming and convex duality, where the latter will play a crucial role in this note.  
As a complete review of literature on this topic is too extensive, we only concentrate on those of immediate interest. 

To our best knowledge, Cvitani\'c and Karatzas \cite{CK96} are the first to apply convex duality to solve the utility maximization problem under proportional transaction costs.
They considered a num\'eraire-based model within the It\^o framework and the agent was assumed to liquidate his portfolio to the bond at the end of trading. 
In \cite{CK96}, the existence of primal solution is ensured only when the dual problem admits a suitable solution. 
In the same setting, Cvitani\'c and Wang \cite{CW01} afterward provided duality results without appealing to such assumption on the dual solution. 
They achieved this by suitably enlarging the domain of the dual problem, as Kramkov and Schachermayer did in \cite{KS99} for investigating a frictionless counterpart. 

In parallel with the num\'eraire-based model, Kabanov \cite{Kab99} introduced a more general multi-currency model based on the concept of solvency cone. 
In this framework, Deelstra et al. gave dual formulation of the multivariate utility maximization in \cite{DPT01}, when the market was associated with a continuous semimartingale of classical no-arbitrage features. 
Thereafter, a similar problem has been considered with random endowment in \cite{Bou02}, in which  liquidation is required. 
For more general market models beyond semimartingales, Campi and Owen \cite{CO11} solved the utility maximization problem with transaction costs by convex duality under existence of consistence price systems (cf.~\cite{CS06}). 
In particular, duality results in \cite{CO11} rely on the enlargement of the dual domain formed by consistent price systems to a set of finitely additive measures, which is based on the idea
of \cite{CSW01, KZ03, OZ09}. 
Notice that this step is altered in the proof of a similar problem in the num\'eraire-based context (cf.~\cite{CS15duality, CSY15}), where the abstract theorem in \cite{KS99} applies, 
  which is owed to the $L^0$-bipolar property between the dual domain defined in terms of supermartingale deflators and the primal one.

The results in \cite{CO11} has been subsequently generalized by Benedetti and Campi in \cite{BC12} to the case with bounded random endowment. 
In this note, we consider a similar problem as in \cite{BC12} however for the num\'eraire-based model rather than the multi-currency one, 
  i.e., we assume that the market consists of one bond and one stock, and the investor has to liquidate all his/her position in stock at the end of trading. 
We emphasize that essentially we goes {\bf no further} than Benedetti and Campi. Indeed, by applying the approach in \cite{CSW01}, 
  we merely present how convex duality works for the problem under transaction costs with the simpler num\'eraire-based model. Moreover, it is observed that every result in \cite{CSW01} has its extension  under transaction costs.
We remark that only utility functions supporting the positive half-plane are concerned in \cite{BC12} as well as in this note.  
For the results on utility functions allowing for negative wealth, we refer the reader to \cite{LY16} (submission in preparation, draft available on request).

The remainder of the article is organized as follows. 
In Section 2, we introduce the financial market model with transaction costs. Moreover, the primal and dual problems are defined. In particular, thanks to the crucial super-replication theorem proved in \cite{Sch14super} (compare also \cite{CS06}), we could enlarge the collection of martingale densities corresponding to consistent price systems to a set of finitely additive measures in a similar manner as in \cite{CSW01}.   
Then, we establish the convex duality results as in \cite{CSW01} by characterizing the primal value function and the primal optimizer with respect to the dual ones in Section 3.

\section{formulation of the problem}

\subsection{The financial model} 

 We consider a model of a financial market which consists of two assets, one bond and one stock. 
 We work in discounted terms, i.e., the price of the bond $B$ is constant and normalized to $B\equiv 1$. 
 We denote by $S=(S_t)_{0\leq t\leq T}$ the price process of the stock, which is based  
 on a filtered probability space $(\Omega, \cF, (\cF_t)_{0\leq t\leq T}, \massP)$ 
 satisfying the usual hypotheses of right continuity and saturatedness, where $\cF_0$ is assumed to be trivial.  
 Here, $T$ is a finite time horizon. Throughout the paper we make the following assumption:
 
 \begin{assumption} \label{Sassumption}
  The process $S=(S)_{0\leq t\leq T}$ is adapted to $(\cF_t)_{0\leq t\leq T}$, with c\`adl\`ag and strictly positive paths. 
 \end{assumption}
  
 We introduce proportional transaction costs $\lambda> 0$ for the trading of the stock. 
 The process $((1-\lambda)S_t,S_t)_{0\leq t\leq T}$ models the bid and ask price of the stock $S$, respectively, 
 which means that the agent has to pay a higher ask price $S_t$ to buy stock shares but only receives a lower bid price $(1-\lambda)S_t$ when selling them. 
 We assume $\lambda<1$ for obvious economic reasons.

We also assume that the agent is endowed with initial wealth $x>0$ and receives an exogenous endowment, 
whose cumulative process is denoted by $e=(e_t)_{0\leq t\leq T}$, $e_0=0$, 
assumed bounded, adapted, with $\rho:=\|e_T\|_{\infty}<\infty$. 
We note that $e_t$ can take negative values, interpreted as mandatory consumption. 
In our case, to solve an expected utility maximization problem, only the final value $e_T$ matters. 

We model \textbf{trading strategies} by $\RR^2$-valued, predictable processes $\varphi=(\varphi^0_t,\varphi^1_t)_{0\leq t\leq T}$ of finite variation, 
 where $\varphi^0_t$ and $\varphi^1_t$ denote the holdings in units of the riskless and the risky asset, respectively, after rebalancing the portfolio at time $t$.
 
To establish our model, we adopt several definitions from \cite{Sch14super} and \cite{Sch14admissible}. 

\begin{definition} 
  A strategy $\varphi = (\varphi_t^0,\varphi_t^1)_{0\leq t\leq T}$ is called {\bf self-financing under transaction costs $\lambda$}, \index{self-financing} if 
   \begin{equation} \label{SF}
     \int_s^td\varphi^0_u \leq -\int_s^tS_ud\varphi_u^{1,\uparrow} + \int_s^t(1-\lambda)S_ud\varphi_u^{1,\downarrow}
   \end{equation}
  for all $0\leq s\leq t\leq T$, where the integrals are defined via
   \begin{align*}
     \int_s^tS_ud\varphi_u^{1,\uparrow} &:= \int_s^tS_ud\varphi_u^{1,\uparrow,c} + \sum_{s<u\leq t}S_{u-}\Delta\varphi_u^{1,\uparrow} + \sum_{s\leq u<t}S_u\Delta_+\varphi_u^{1,\uparrow}, \\
     \int_s^tS_ud\varphi_u^{1,\downarrow} &:= \int_s^tS_ud\varphi_u^{1,\downarrow,c} + \sum_{s<u\leq t}S_{u-}\Delta\varphi_u^{1,\downarrow} + \sum_{s\leq u<t}S_u\Delta_+\varphi_u^{1,\downarrow}.
   \end{align*}  
\end{definition}

 The self-financing condition \eqref{SF} states that purchases and sales of the risky asset are accounted for in the riskless position: 
 \begin{align*} 
   d\varphi^{0,c}_t&\leq-S_td\varphi^{1,\uparrow,c}_t+(1-\lambda)S_td\varphi^{1,\downarrow,c}_t, \notag \\
   \Delta\varphi^0_t&\leq-S_{t-}\Delta\varphi^{1,\uparrow}_t +(1-\lambda)S_{t-}\Delta\varphi^{1,\downarrow}_t,\notag   \\
   \Delta_+\varphi^0_t&\leq-S_t\Delta_+\varphi^{1,\uparrow}_t +(1-\lambda)S_t\Delta_+\varphi^{1,\downarrow}_t, 
 \end{align*}
 for $0\leq t\leq T$.



\begin{definition}
 A self-financing strategy $\varphi$ is \textbf{admissible}, if its liquidation value
 \begin{eqnarray*}
   V_t^{\mathrm{liq}}(\varphi):= \varphi_t^0 + (\varphi_t^1)^+(1-\lambda)S_t - (\varphi_t^1)^-S_t 
                 \geq -M, \quad  a.s., 
 \end{eqnarray*}
for some $M>0$, simultaneously for all $t\in[0,T]$. 
\end{definition}

For $x\in\RR$, we denote by $\cA^{\lambda}_{adm}(x)$ the set of all admissible self-financing trading strategies under transaction costs $\lambda$ 
 with $(\varphi_{0}^0, \varphi_{0}^1)=(x, 0)$ and $\varphi_{T}^1 = 0$ and
 $$ \cC^\lambda(x):= \left\{V^{\mathrm{liq}}_T(\varphi)\,\Big|\, \varphi = (\varphi^0,\varphi^1)\in\cA^{\lambda}_{adm}(x) \right\}. $$
As explained in \cite[Remark 4.2]{CS06}, we can assume without loss of generality that $\varphi^1_T=0$ and therefore
 $$ \cC^\lambda(x) = \left\{\varphi_T^0 \,\Big|\, \varphi = (\varphi^0,\varphi^1)\in\cA^{\lambda}_{adm}(x)\right\}. $$
Note that the restriction on trading strategies $\varphi_{T}^1 = 0$ means that all stock shares are liquidated at time $T$, 
 i.e., a trading strategy must begin and end with a cash position only.  

\vspace{3mm}

To ensure the optimization problem meaningful, the assumption of the absence of arbitrage, is required here too.
We recall some useful results of the arbitrage theory in markets with transaction costs.

\begin{definition}
  Fix $0<\lambda <1$ and a price process $S=(S_t)_{0\le t\le T}$ as above.
  A $\lambda$-{\bf consistent price system} is a two dimensional strictly positive process $Z=(Z^0_t,Z^1_t)_{0\le t\le T}$ 
      with $Z^0_0=1$, that consists of a martingale $Z^0$ and a (local) martingale $Z^1$ under $\massP$ such that 
      \begin{equation}\label{J10}
       \widetilde{S}_t:=\frac{Z^1_t}{Z^0_t} \in [(1-\lambda)S_t, S_t],\qquad a.s. 
      \end{equation}
      for $0\leq t\leq T$.

  We denote by $\mathcal{Z}^{\lambda}_e(S)$ the set of $\lambda$-consistent price systems. 
  
  We say that S satisfies $(CPS^{\lambda})$,
  if there is a consistent price system for given transaction costs $\lambda \in (0,1)$.
\end{definition}

   \begin{remark}
     In the above definition, $Z^0$ defines a density process of an equivalent (local) martingale measure $\massQ\sim\massP$ for a price process $\widetilde{S}$ 
      evolving in the bid-ask spread $[(1-\lambda)S,S]$, and $Z^1=Z^0\widetilde{S}$.
   \end{remark}

In the context with transaction costs, the consistent price system plays the same role as the equivalent localmartingale measure in frictionless financial markets.  
To issue the important superreplication theorem, we have the following assumption throughout the paper:

\begin{assumption}  \label{CPSassumption}
 $S$ satisfies $(CPS^{\mu})$ for all $\mu\in(0,1)$. 
\end{assumption}

\begin{theorem}[Superreplication theorem] \label{superreplication}
  Let $S$ satisfy Assumption \ref{Sassumption} and Assumption \ref{CPSassumption}. Fix $0<\lambda<1$. 
  Let $ g\in L^0(\Omega,\cF,\massP)$ be a random variable bounded from below, i.e., $g\geq -M$ almost surely for some $M>0$. 
  
  Then $g\in\cC^{\lambda}(x)$, i.e., $(g,0)$ is the terminal value of some $\lambda$-self-financing, admissible trading strategy $(\varphi^0_t,\varphi^1_t)_{0\leq t\leq T}\in\cA^{\lambda}_{adm}(x)$, 
  if and only if $$ \massE[Z^0_Tg] \leq x, $$
  for every $\lambda$-consistent price system $(Z^0,Z^1)$.
\end{theorem} 

\begin{proof}
 See \cite[Theorem 1.4]{Sch14super}. 
\end{proof}


\subsection{Optimization problem}
Now suppose the agent's preferences over terminal wealth are modeled by a utility function
$U:(0,\infty)\to\RR$, which is strictly increasing, strictly concave, continuously differentiable and 
satisfies the Inada condition:
$$ U'(0) := \lim_{x\to 0}U'(x) = \infty \quad\textnormal{ and } \quad U'(\infty) := \lim_{x\to \infty}U'(x) = 0. $$
Without loss of generality, we may assume $U(\infty)>0$ to simplify the analysis. 
Define also $U(x)=-\infty$ whenever $x\leq 0$.

\begin{assumption} \label{U(x)assumption}
 The utility function $U$ satisfies the reasonable asymptotic elasticity, i.e.
 $$ AE(U):= \limsup_{x\to\infty}\frac{xU'(x)}{{U}(x)}< 1. $$
\end{assumption}

For financial interpretation and more results about the previous assumption, we refer to \cite{KS99}.

Then, we restrict our attention to the terminal liquidation wealth.
For $x>0$, the primal problem is to maximize the expected utility function from terminal wealth
 $$ u(x):=\sup_{(\varphi^0,\varphi^1) \in \cA^{\lambda}_{adm}(x)}\massE[U(\varphi_T^0 +  e_T)]. $$
We denote $\cC^\lambda:=\cC^\lambda(0)$. 
Note that $\cC^\lambda(x)=x+\cC^\lambda$, therefore the above problem may also be written as 
 \begin{equation} \label{primal_problem}
  u(x):=\sup_{g\in\widetilde{\cC}^{\lambda}}\massE[U(x+g+e_T)],
 \end{equation}
 where the set $\widetilde{\cC}^{\lambda}$ consists of those elements of $\cC^{\lambda}$ for which the above expectation is well defined. 

Finally, in order to exclude trivial case, we have the following assumption: 
\begin{assumption} \label{u(x)assumption}
  The value function $u(x)$ is finitely valued for some $x>\rho$. 
\end{assumption}

The concavity of $u(x)$ and Assumption \ref{u(x)assumption} imply that $u(x)<\infty$ for all $x\in\RR$.

\subsection{Dual problem}

Let us denote $V: \RR_+\to\RR$ the convex conjugate function of $U(x)$ defined by 
$$ V(y):=\sup_{x>0}\{U(x)-xy\}, \quad y>0. $$
From classical results of convex analysis, we know that 
$V(y)$ is strictly decreasing, strictly convex and continuously differentiable and satisfies 
$$ V(0)=U(\infty), \quad V(\infty)=U(0). $$

We also define $I:(0,\infty)\to (0,\infty)$ the inverse function of $U'$ on $(0,\infty)$, 
which is strictly decreasing, and satisfies $I(0)=\infty$, $I(\infty)=0$ and $I = -V'$.

 For a treatment of the problem at hand, the usual dual space
  $$ \cM^{\lambda}:= \left\{Z_T^0\in L^1(\massP) \,\Big| \, (Z^0,Z^1)\in\cZ^{\lambda}_e(S) \right\}, $$
   which is a subset of $L^1$, is too small. 
 As in \cite{CSW01}, we extend the usual domain to $ba=(L^{\infty})^{*}$, the dual space of $L^{\infty}$ and define the following subset of $ba$, 
   which is equipped with the weak-star topology $\sigma(ba,L^{\infty})$,
  $$ \cD^{\lambda}:=\big\{Q\in ba\, \big| \, \|Q\|=1 \,\textnormal{ and }\, \langle Q, g \rangle\leq 0\, \textnormal{ for all }\, g\in\cC^{\lambda}\cap L^{\infty} \big\},  $$
  and $\cD^{\lambda,r}:=\cD^{\lambda}\cap L^1$, where $r$ stands for regular. 
 
 \begin{remark}
   $\cD^{\lambda}$ is clearly convex and also $\sigma(ba,L^{\infty})$-compact by Alaoglu's theorem.
 \end{remark}

 \begin{remark}
   It is easy to see that $- L^{\infty}_+ \subseteq \mathcal C^{\lambda}$, then $\cD^{\lambda}\subseteq ba_+$, hence $\mathcal D^{\lambda,r} \subseteq L^1_+$.
 \end{remark}

 \begin{remark}
   By Theorem \ref{superreplication}, each $g\in\cC^{\lambda}$ satisfies $\massE[Z^0_Tg]\leq 0$, for every consistent price system $(Z^0,Z^1)$, so $\cM^{\lambda}\neq \emptyset$.
   Since $\cM^{\lambda}\subseteq \cD^{\lambda,r}\subseteq \cD^{\lambda}$, the sets $\cD^{\lambda}$ and $\cD^{\lambda,r}$ are nonempty.
 \end{remark}
 
 \begin{lemma}  \label{D_is_sigma_closure_of_M}
   The set $\cD^{\lambda}$ is the $\sigma(ba,L^{\infty})$-closure of $\cM^{\lambda}$. 
 \end{lemma}
   
 \begin{proof}
   It is clear that $\cM^{\lambda}\subseteq\cD^{\lambda}$ and $\cD^{\lambda}$ is $\sigma(ba,L^{\infty})$-closed, 
     hence $$\overline{\cM^{\lambda}}^{\sigma(ba,L^{\infty})}\subseteq\cD^{\lambda}.$$ 
     
   Assume now that there exists an element $\widetilde{Q}\in\cD^{\lambda}$ satisfying $\widetilde{Q}\notin\overline{\cM^{\lambda}}^{\sigma(ba,L^{\infty})}$. 
   As $\cM^{\lambda}$ is a convex set, the $\sigma(ba,L^{\infty})$-closure $\overline{\cM^{\lambda}}^{\sigma(ba,L^{\infty})}$ is also convex. 
   By the Hahn-Banach theorem, there exists $f\in L^{\infty} = \big(ba,\sigma(ba,L^{\infty})\big)^*$, such that $\langle\widetilde{Q},f\rangle>\alpha$ and 
     $$ \langle Q,f\rangle \leq \alpha, \qquad \forall Q\in\overline{\cM^{\lambda}}^{\sigma(ba,L^{\infty})}, $$
     for some $\alpha\in\RR$. 
   In particular, $\massE[Z_T^0f]\leq \alpha$ for all $Z^0_T\in\cM^{\lambda}$, which follows by Theorem \ref{superreplication} that $f\in\cC^{\lambda}(\alpha)$, 
     therefore $f-\alpha\in\cC^{\lambda}$. 
   By the definition of $\cD^{\lambda}$, we obtain that 
     $$\langle \widetilde{Q},f-\alpha\rangle = \langle \widetilde{Q},f\rangle - \alpha \leq 0, $$
     which contradicts the fact that $\langle\widetilde{Q},f\rangle>\alpha$. 
 \end{proof}
 
 \begin{lemma}  \label{XincC}
   Let $g\in L^{\infty}$. 
   Then $g\in\cC^{\lambda}$ if and only if $\langle Q,g\rangle \leq 0$ for all $Q\in\cD^{\lambda,r}$.
 \end{lemma}
 
 \begin{proof}
  The necessity follows directly from the definition of $\cD^\lambda$. 
  The sufficiency follows from Theorem \ref{superreplication}, since $\cM^{\lambda}\subseteq \cD^{\lambda,r}\subseteq \cD^\lambda$.
 \end{proof}

The following proposition collects some properties of the space $ba_+$; more information can be found in Appendix of \cite{CSW01} and references there. 

\begin{proposition} \label{propertiesofLinfty*} 
\quad 
 \begin{enumerate}
  \item The set $ba_+$ can be identified as the set of all nonnegative finitely additive bounded set functions on $\cF$, which vanish on the $\massP$-null sets. 
  \item Every $Q\in ba_+$ admits a unique decomposition in the form of $Q = Q^r+Q^s$,  
        where the regular part $Q^r$ is the maximal countably additive measure on $\cF$, that is dominated by $Q$, and 
        the singular part $Q^s$ is purely finitely additive and does not dominate any nontrivial countably additive measure. 
  \item $Q\in ba_+$ is purely finitely additive, i.e., $Q^r=0$, if and only if for every $\varepsilon>0$, 
        there exists a set $A_{\varepsilon}\in\cF$ such that $\massP(A_{\varepsilon})>1-\varepsilon$ and $Q(A_{\varepsilon})=0$. 
  \item Suppose $(Q_n)_{n\in\NN}\subseteq ba_+$ is a sequence such that $\frac{dQ_n^r}{d\massP}\to f$ almost surely for some $f\geq 0$. 
        Then any weak-star cluster point $Q$ of $(Q_n)_{n\in\NN}$ satisfies $\frac{dQ^r}{d\massP}=f$ almost surely. 
 \end{enumerate}
\end{proposition}

For any $Q\in ba_+$, we may define 
  $$ \langle Q,X\rangle:=\lim_{n\to\infty}\langle Q,X\wedge n\rangle\in [0,\infty], $$
 for all $X\in L_+^0$. 
For $X\in L^0$, set $\langle Q,X\rangle=\langle Q,X^+\rangle-\langle Q,X^-\rangle$ whenever this is well-defined. 
We observe that each $ g\in\cC^{\lambda}$ is uniformly bounded from below and thus, $\langle Q,g\rangle\leq 0$, for all $g\in\cC^{\lambda}$ and $Q\in\cD^{\lambda}$.

Now we define the dual optimization problem by 
\begin{equation} \label{defvy}
  v(y):=\inf_{Q\in\cD^{\lambda}}J(y,Q):= \inf_{Q\in\cD^{\lambda}}\left\{\massE\left[V\left(y\frac{dQ^r}{d\massP}\right)\right]+y\langle Q,e_T\rangle \right\}.
\end{equation}


\section{Main Theorem}

In the following theorem, we see that even by adding transaction costs, the results are similar as in \cite{CSW01}.
Now we state the main result: 

\begin{theorem}
  Under the Assumption \ref{Sassumption}, \ref{CPSassumption}, \ref{U(x)assumption}, \ref{u(x)assumption},
  we have
  \begin{enumerate}
   \item $u(x)<\infty$ for all $x\in\RR$ and $v(y)<\infty$ for all $y>0$.
   \item The primal value function is continuously differentiable on $(x_0,\infty)$ and $u(x)=-\infty$ for all $x<x_0$, 
         where $x_0:=-v'(\infty)=\sup_{Q\in\cD^\lambda}\langle Q,-e_T\rangle$.
         The dual value function $v$ is continuously differentiable on $(0,\infty)$. 
   \item The functions $u$ and $v$ are conjugate in sense that 
        \begin{equation}   \label{v=u-xy}
            v(y) = \sup_{x>x_0}\{u(x)-xy\}, \quad y>0,
        \end{equation}
        \begin{equation}   \label{u=v+xy}
            u(x) = \inf_{y>0}\{v(y)+xy\}, \quad x>x_0.
        \end{equation}
   \item For all $y>0$, there exists a solution $\widehat{Q}_y\in\cD^\lambda$ to the dual problem, which is unique up to the singular part. 
         For all $x>x_0$, $\widehat{g}:= I\Big(\widehat{y}\frac{d\widehat{Q}^r_{\widehat{y}}}{d\massP}\Big)-x-e_T$ is the solution to the primal problem, 
         where $\widehat{y}=u'(x)$, which attains the infimum of $\{v(y)+xy\}$. 
  \end{enumerate}
\end{theorem}

The rest of this section is devoted to the proof of the above main theorem. 
We split the proof in several lemmas and propositions, where we may see the use of the required assumptions for each step. 


\begin{lemma}   \label{u<infv+xy}
  For all $x\in\RR$,
  $$u(x)\leq \inf_{y>0}\inf_{Q\in\cD^\lambda}\{J(y,Q)+xy\} = \inf_{y>0}\{v(y)+xy\}.$$
\end{lemma}

\begin{proof}
  For the case $x+g+e_T\leq 0$ on a measurable set $A \in \cF $ with $\massP(A)>0$, we get $u(x)=-\infty$, therefore the assertion satisfies trivially. 
  We only have to consider the case $x+g+e_T>0$ $\massP$-a.s. 
  As $g$ is bounded from below by $-(x+\rho)$ and $S$ satisfies $(CPS^{\mu})$ for all $\mu\in(0,1)$, it follows by \cite[Theorem 1]{Sch14admissible} that 
    $g$ can be attained by some $(x+\rho)$-admissible, self-financing trading strategy. 
    
  From the definition of $V(y)$, positivity of $x+g+e_T$, and $\langle Q,g\rangle\leq 0$, it follows
   \begin{equation} \label{UleqJ+xy}
     \begin{aligned}
      \massE[U(x+g+e_T)] &\leq \massE\left[V\left(y\frac{dQ^r}{d\massP}\right) + y\frac{dQ^r}{d\massP}(x+g+e_T)\right] \\
                         &\leq \massE\left[V\left(y\frac{dQ^r}{d\massP}\right)\right] + y\langle Q,x+g+e_T\rangle  \\
                         &\leq \massE\left[V\left(y\frac{dQ^r}{d\massP}\right)\right] + y\langle Q,e_T\rangle + xy \\
                         &= J(y,Q)+xy      
     \end{aligned}
   \end{equation}
  for all $y>0$, $g\in\widetilde{\cC}^{\lambda}$, $Q\in\cD^\lambda$. 
  Taking supremum and infimum at left-and right-hand side, respectively, we obtain the assertion. 
\end{proof}


We now study the dual value function. 

\begin{lemma}
 The function $v(y)$ is finitely valued, for all $y>0$. 
\end{lemma}

\begin{proof}
 By Jensen's inequality, the fact that $V$ is decreasing and $\massE\big[\frac{d\massQ^r}{d\massP}\big]\leq 1$, we have
  \begin{equation} \label{v(y)geqV(y)-y}
    \begin{aligned}
      v(y) &= \inf_{Q\in\cD^\lambda}\left\{\massE\left[V\left(y\frac{dQ^r}{d\massP}\right)\right]+y\langle Q,e_T\rangle \right\}  \\ 
           &\geq \inf_{Q\in\cD^\lambda}V\left(y\massE\left[\frac{dQ^r}{d\massP}\right]\right) -y\rho 
            \geq V\left(y\right) -y\rho   
            > -\infty    
    \end{aligned}
  \end{equation}
 for all $y>0$. 
 
 To show $v(y)<\infty$, we need to recall the duality result without random endowment in \cite{CS15duality} (cf.~in \cite[Theorem 3.2]{CS15duality}). 
 To adapt the setting in that article, we denote by $\widetilde{u}(x)$ and $\widetilde{v}(y)$ be the primal and dual value function, respectively, i.e., 
  \begin{equation*}
   \begin{aligned}
    & \widetilde{u}(x):= \sup_{g\in\widetilde{\cC}^{\lambda}}\massE[U(x+g)], \\
    & \widetilde{v}(y):=\inf_{Z_T^0\in\cM^{\lambda}}\massE\left[V\left(yZ^0_T\right)\right].   
   \end{aligned}
  \end{equation*}
 
 By Assumption \ref{u(x)assumption}, we obtain 
 \begin{equation}
   \widetilde{u}(x) \leq \sup_{g\in\widetilde{\cC}^{\lambda}}\massE[U(x+g+\rho+e_T)] = u(x+\rho) < \infty,
 \end{equation}
 for all $x>0$.
 On the other hand, by \cite[Theorem 3.2]{CS15duality},  
  $$ \widetilde{v}(y) = \sup_{x>0}\{\widetilde{u}(x)-xy\}=\widetilde{u}(\widehat{x}_y)-\widehat{x}_yy<\infty, $$
 for all $y>0$. 
 It follows from
  \begin{eqnarray*}
   v(y) &=& \inf_{Q\in\cD^\lambda}\left\{\massE\left[V\left(y\frac{dQ^r}{d\massP}\right)\right]+y\langle Q,e_T\rangle \right\} \\
        &\leq& \min_{Z^0_T\in\cM^\lambda}\massE\left[V\left(yZ_T^0\right)\right]+y\rho  \\
        &=& \widetilde{v}(y) + y\rho,  
 \end{eqnarray*}
 that $v(y)<\infty$, for all $y>0$. 
\end{proof}

\begin{lemma}
 For any $y>0$, the infimum of the left-hand side of \eqref{defvy} is attained by some $\widehat{Q}_y\in\cD^\lambda$. 
\end{lemma}

\begin{proof}
 Let $(Q_n)_{n\in\NN}\subseteq \cD^\lambda$ be the minimizing sequence, i.e. 
 $$ v(y)=\lim_{n\to\infty}\left\{\massE\left[ V\left(
 y\frac{dQ_n^r}{d\massP}\right)
 \right]+y\langle Q_n,e_T \rangle  \right\}. $$
 
Since $\cD^\lambda$ is convex and $\big(\frac{dQ^r_n}{d\massP}\big)_{n\in\NN}$ is $L^1$-bounded, 
 we can find a sequence $(\widetilde{Q}_n)_{n\in\NN}$ with $\widetilde{Q}_n\in\textnormal{conv}(Q_k;k\geq n)$
 such that $\frac{d\widetilde{Q}_n^r}{d\massP}$ converges almost surely to some $f\geq 0$. 
 
 Clearly $\abs{\langle \widetilde{Q}_n,e_T\rangle}\leq \rho $. 
 Then we can extract a subsequence of $\widetilde{Q}_n$, which is still denoted by $\widetilde{Q}_n$, 
 such that $\langle \widetilde{Q}_n,e_T\rangle$ converges. 
 
Note that $\cD^\lambda$ is $\sigma(ba,L^\infty)$-compact, 
thus the sequence $(\widetilde{Q}_n)_{n\in\NN}$ has a cluster point $\widehat{Q}_y\in\cD^\lambda$. 
 From Proposition \ref{propertiesofLinfty*} (4) we have 
 $$ \frac{d\widehat{Q}^r_y}{d\massP} = f = \lim_{n\to\infty}\frac{d\widetilde{Q}^r_n}{d\massP}. $$
 
 Similarly to \cite[Lemma 3.2]{KS99}, we obtain the uniform integrability of $\left\{V^-\left(y\frac{d\widetilde{Q}_n^r}{d\massP}\right)\right\}_{n\in\NN}$. 
 By Fatou's Lemma, we have
 $$ \liminf_{n\to\infty}\massE\left[V\left( y\frac{d\widetilde{Q}_n^r}{d\massP} \right)\right] \geq \massE\left[V\left( y\frac{d\widehat{Q}_y^r}{d\massP} \right)\right]. $$
 Since $\langle \widehat{Q}_y, e_T\rangle$ is a cluster point of $(\langle \widetilde{Q}_n, e_T\rangle)_{n\in\NN}$, which converges, we have 
 $$ \langle \widehat{Q}_y, e_T\rangle = \lim_{n\to\infty} \langle \widetilde{Q}_n, e_T\rangle. $$
 Hence,
 \begin{eqnarray*}
   J(y,\widehat{Q}_y) &=& \massE\left[V\left(y\frac{d\widehat{Q}_y^r}{d\massP}\right)\right] + y\langle \widehat{Q}_y,e_T \rangle  \\
                      &\leq& \liminf_{n\to\infty}\left\{ \massE\left[V\left( y\frac{d\widetilde{Q}_n^r}{d\massP} \right)\right] + y\langle \widetilde{Q}_n,e_T \rangle \right\} \\
                      &\leq& \lim_{n\to\infty} \left\{ \massE\left[V\left( y\frac{dQ_n^r}{d\massP} \right)\right] + y\langle Q_n,e_T \rangle \right\} \\
                      &=& v(y), 
 \end{eqnarray*}
 which gives the optimality of $\widehat{Q}_y\in\cD^\lambda$.
\end{proof}

\begin{lemma}
 The solution of the dual problem might not be unique, but its countably additive part is unique. 
\end{lemma}

\begin{proof}
 Assume that $Q_1$ and $Q_2$ are two minimizers such that $Q_1^r\neq Q_2^r$. 
 Let $Q:=\frac{1}{2}Q_1+\frac{1}{2}Q_2\in\cD^\lambda$. 
 By the strict convexity of $V$, 
 $$ \massE\left[V\left( y\frac{dQ^r}{d\massP} \right)\right]<\frac{1}{2}\massE\left[V\left( y\frac{dQ^r_1}{d\massP} \right)\right]+\frac{1}{2}\massE\left[V\left( y\frac{dQ^r_2}{d\massP}\right)\right],$$
 hence, 
 $$ J(y,Q)< \frac{1}{2}J(y,Q_1)+\frac{1}{2}J(y,Q_2)=J(y,\widehat{Q}_y), $$
 which is in contradiction to the optimality of $\widehat{Q}_y$.
\end{proof}

\begin{lemma} \label{convexvy}
 The dual value function $v(\cdot)$ is strictly convex .
\end{lemma}

\begin{proof}
 It follows directly from the strict convexity of the function $V$.  
\end{proof}

\begin{proposition} \label{UI}
 For all $y>0$, $ \frac{d{\widehat Q}_y^r}{d\massP} I \Big( (y-\varepsilon)  \frac{d{\widehat Q}_y^r}{d\massP} \Big) $
 is uniformly integrable for  sufficiently small  $\varepsilon >0$.            
\end{proposition}

To prove this proposition, we recall a result from \cite{KS99}.

\begin{lemma} \label{yI<VandV<CV}
 Under Assumption \ref{U(x)assumption}, there exist $y_0 > 0$ and $0 < \gamma < 1$  
 such that
$$ yI(y) < \frac{ \gamma}{1-\gamma}V(y) \quad \textnormal{and} \quad V(\beta y) < \beta^{-\frac{\gamma}{1-\gamma}}V(y) \quad $$
for all $0< y < y_0$ and $0<\beta<1$.
\end{lemma}

\begin{proof}[Proof of Proposition \ref{UI}]
 By Lemma \ref{yI<VandV<CV}, we can find a $y_0 > 0$, such that, for all $0< y < y_0$ and sufficiently small $\varepsilon>0$,
  \begin{equation*}
    \begin{aligned}
     0 &\leq\frac{d{\widehat Q}_y^r}{d\massP} I\left((y-\varepsilon)\frac{d{\widehat Q}_y^r}{d\massP} \right)
                 {\boldsymbol{1}}_{\left\{y\frac{d{\widehat Q}_y^r}{d\massP}  < y_0\right\}}  \\
       &= \frac{1}{y-\varepsilon}\frac{y-\varepsilon}{y}y \frac{d{\widehat Q}_y^r}{d\massP} I\left(\frac{y-\varepsilon}{y}y\frac{d{\widehat Q}_y^r}{d\massP} \right)
                 {\boldsymbol{1}}_{\left\{y\frac{d{\widehat Q}_y^r}{d\massP}  < y_0\right\}}  \\
       &\leq \frac{1}{y-\varepsilon}\frac{\gamma}{1-\gamma} V\left(\frac{y-\varepsilon}{y}y\frac{d{\widehat Q}_y^r}{d\massP}\right){\boldsymbol{1}}_{\left\{y\frac{d{\widehat Q}_y^r}{d\massP} < y_0\right\}} \\
       &\leq \frac{ \gamma C}{(y- \varepsilon)(1-\gamma)} \abs{V\left( y  \frac{d{\widehat Q}_y^r}{d\massP} \right)}, 
    \end{aligned}
  \end{equation*}
 where $C = \big(\frac{y-\varepsilon}{y}\big)^{-\frac{\gamma}{1-\gamma}}$.
 Since $I$ is decreasing and positive,
\begin{eqnarray*} 
  0\leq \frac{d{\widehat Q}_y^r}{d\massP} I\left( (y-\varepsilon)  \frac{d{\widehat Q}_y^r}{d\massP} \right){\boldsymbol{1}}_{\left\{y \frac{d{\widehat Q}_y^r}{d\massP} \geq y_0\right\}}
    \leq \frac{d{\widehat Q}_y^r}{d\massP}I\left( \frac {y-\varepsilon}{y} y_0 \right).
 \end{eqnarray*}
Therefore,
$$ 0\leq\frac{d{\widehat Q}_y^r}{d\massP} I \left( (y-\varepsilon)  \frac{d{\widehat Q}_y^r}{d\massP} \right)
    \leq K\abs{V\left( y  \frac{d{\widehat Q}_y^r}{d\massP}\right)} + \frac{d{\widehat Q}_y^r}{d\massP}I\left(\frac{y_0}{2}\right), $$
for some constant $K>0$. 
Since the right-hand side is an element in $L^1$, we obtain the uniform integrability of $ \frac{d{\widehat Q}_y^r}{d\massP} I \Big( (y-\varepsilon)  \frac{d{\widehat Q}_y^r}{d\massP} \Big) $
for sufficiently small $\varepsilon>0$.
\end{proof}

\begin{lemma} \label{v'=-QI+Qe}
 The dual value function is continuously differentiable on $(0,\infty)$, 
 $$ v'(y)= -\left\langle \widehat{Q}_y^r, I\left(y\frac{d\widehat{Q}_y^r}{d\massP}\right)\right\rangle + \langle \widehat{Q}_y, e_T\rangle. $$
\end{lemma}

\begin{proof}
  Let $y >0$ be arbitrary. 
  Define $$f(z) := \massE\left[V\left(z\frac{d{\widehat Q}_y^r}{d\massP} \right) \right] + z\left\langle {\widehat Q}_y, e_T  \right\rangle. $$ 
  It is easy to see that $f(z)$ is convex, $f(\cdot) \geq v(\cdot)$ and $f(y) = v(y)$, which implies that 
         $\triangle^- f(y) \leq \triangle^- v(y) \leq \triangle^+ v(y) \leq \triangle^+ f(y)$,
  where $\triangle^{\pm}$ describe the left and the right derivatives, respectively.

  By the convexity of $V(\cdot)$ and the Fatou's lemma, it follows that 
   \begin{equation*}
     \begin{aligned}
      \triangle^+ f(y) &\leq \limsup_{\varepsilon\to 0}\frac{1}{\varepsilon}\massE\left[V\left((y+\varepsilon)\frac{d{\widehat Q}_y^r}{d\massP}\right)
                                                                                   - V\left(y\frac{d{\widehat Q}_y^r}{d\massP}\right)\right]
                                                 + \left\langle {\widehat Q}_y, e_T  \right\rangle \\  
                   &\leq \limsup_{\varepsilon\to 0}\massE\left[\frac{d{\widehat Q}_y^r}{d\massP}V'\left((y+\varepsilon)\frac{d{\widehat Q}_y^r}{d\massP}\right) \right]
                                                 + \left\langle {\widehat Q}_y, e_T  \right\rangle \\ 
                   &\leq \massE\left[\frac{d{\widehat Q}_y^r}{d\massP}V'\left(y\frac{d{\widehat Q}_y^r}{d\massP}\right) \right]
                                                 + \left\langle {\widehat Q}_y, e_T  \right\rangle \\                           
                    &= - \left\langle {\widehat Q}_y^r,  
                                         I \left( y  \frac{d{\widehat Q}_y^r}{d\massP}\right) \right\rangle
                        + \left\langle {\widehat Q}_y, e_T  \right\rangle. 
     \end{aligned}
   \end{equation*}
On the other side, by Proposition \ref{UI}, we can apply Fatou's lemma again, and it follows that
 \begin{equation*} 
  \begin{aligned}
    \triangle^- f(y) &\geq \liminf_{\varepsilon\to 0}\massE\left[-\frac{d{\widehat Q}_y^r}{d\massP}I\left((y-\varepsilon)\frac{d{\widehat Q}_y^r}{d\massP}\right)\right]+\left\langle{\widehat Q}_y, e_T\right\rangle \\
                     &\geq - \left\langle {\widehat Q}_y^r, I\left(y\frac{d{\widehat Q}_y^r}{d\massP}\right)\right\rangle + \left\langle\widehat{Q}_y, e_T\right\rangle.
  \end{aligned}
 \end{equation*}
Thus, $\triangle^- f(y) = \triangle^- v(y) =v'(y) = \triangle^+ v(y) = \triangle^+ f(y)$.

  By strict convexity, $v(\cdot)$ is continuously differentiable. 
\end{proof}

\begin{lemma}  \label{inadav}
In particular,
 $$v'(0+)=-\infty, \quad v'(\infty)\in\left[\inf_{Q\in\cD^\lambda}\langle Q,e_T\rangle, \sup_{Q\in\cD^\lambda}\langle Q,e_T\rangle \right].$$
\end{lemma}

\begin{proof}
 From \eqref{v(y)geqV(y)-y}, we have $v(0+) \geq V(0+)$.
 On the other hand, by the definition of $v(\cdot)$ and the decrease of $V(\cdot)$, we have that, for any $Q\in\cD^\lambda$,
  \begin{equation*} 
    v(y) \leq \massE\left[V\left(y\frac{d{ Q}^r}{d\massP}\right)\right]+ y\left\langle {Q}, e_T \right\rangle 
         \leq V(0+) + y \rho , 
  \end{equation*}
  which implies $v(0+) \leq V(0+)$.
 Hence $v(0+) = V(0+) = U(\infty)$.
 We only need to consider the case that $U(\infty) < \infty$, indeed, if $U(\infty) = \infty$, we get $v(0+) = \infty$, and it follows trivially $v'(0+) = -\infty$.

 By the convexity of $v$ and $V$, \eqref{v(y)geqV(y)-y}, we have 
  \begin{equation*} 
   \begin{aligned}
     v'(0+) &\leq \frac{v(y)-v(0+)}{y} \leq \frac{\massE\left[V\left(y\frac{d{Q}^r}{d\massP}\right)-V(0+)\right] + y\rho}{y} 
             \leq -\massE\left[\frac{d{Q}^r}{d\massP} I\left(y\frac{d{ Q}^r}{d\massP}\right)\right] + \rho,
   \end{aligned}
  \end{equation*}
  for all $y > 0 $ and $Q \in \mathcal D^\lambda$. 
 Letting $y\to 0$, we obtain $v'(0+) = -\infty$ by monotone convergence theorem. 

 By the definition of $v(\cdot)$ and l'H\^opital's rule, we have
 \begin{equation*}
  \begin{aligned}
    v'(\infty) &= \lim_{y\to\infty}\frac{v(y)}{y}
                = \lim_{y\to\infty}\frac{\inf_{Q\in\cD^\lambda}\left\{\massE\left[V\left(y\frac{d{Q}^r}{d\massP}\right)\right] + y\left\langle {Q}, e_T\right\rangle\right\}}{y}  \\
               &\in \left[K +\inf_{Q\in\cD^\lambda}\langle Q,e_T\rangle, K + \sup_{Q\in\cD^\lambda}\langle Q,e_T\rangle  \right],
  \end{aligned}
 \end{equation*}
 where $$ K =  \lim_{y\to\infty}\frac{1}{y}\inf_{Q\in\cD^\lambda}\massE\left[V\left(y\frac{d{ Q}^r}{d\massP}\right)\right]. $$

Since $-V(\cdot)$ is increasing and $I(y)\to 0$ as $y\to\infty$, we have that for all $\varepsilon>0$, there exists $C_{\varepsilon}>0$ 
such that 
$$ -V(y)\leq C_{\varepsilon}+\varepsilon y,  $$
for all $y>0$. Hence 
$$ 0\leq -K = \lim_{y\to\infty}\frac{\sup_{Q\in\cD^\lambda}\massE\left[-V\left(y\frac{d{ Q}^r}{d\massP}\right)\right]}{y} 
          \leq \lim_{y\to\infty}\frac{C_{\varepsilon}+\varepsilon y}{y} = \varepsilon. $$
Consequently, $K=0$ and the claim follows. 
\end{proof}

%

Now let us consider the next step, $\inf_{y>0}\{v(y)+xy\}$:

If $x<x_0:=-v'(\infty)$ we have $v'(y)+x<0$ for all $y>0$, hence $\inf_{y>0}\{v(y)+xy\}=-\infty$ and by Lemma \ref{u<infv+xy} we have 
 $$ u(x)\leq \inf_{y>0}\{v(y)+xy\} =-\infty. $$
In this case the optimization problem is trivial. 

For each $x>x_0$, there exists a unique $\widehat{y}>0$, such that $v'(\widehat{y})+x=0$, and $\widehat{y}$ attains the infimum of $\{v(y)+xy\}$.
After having shown the existence of optimizer of the dual problem, we come back to the primal problem. 
For simplicity, denote ${\widehat Q}:={\widehat Q}_{\widehat y}$.
Let us consider $$\widehat{g}:=I\left(\widehat{y}\frac{d\widehat{Q}^r}{d\massP}\right)-x-e_T. $$
Since $I(\cdot)$ is positive, we have that $x+\widehat{g}+e_T>0$ $\massP$-a.s.
It follows from Lemma \ref{v'=-QI+Qe} 
\begin{equation} \label{-x=-Qr,x+X+QseT}
 \begin{aligned}
  -x = v'({\widehat y}) &= -\left\langle \widehat{Q}^r, I\left(\widehat{y}\frac{d\widehat{Q}^r}{d\massP}\right)\right\rangle + \left\langle \widehat{Q}, e_T  \right\rangle \\
                    &= -\left\langle \widehat{Q}^r, x+ \widehat{g} +e_T\right\rangle + \left\langle \widehat{Q}, e_T  \right\rangle   \\
                    &= -\left\langle \widehat{Q}^r, x+ \widehat{g}\right\rangle + \left\langle \widehat{Q}^s, e_T  \right\rangle.  
 \end{aligned}
\end{equation}

The following lemmas will show that $\widehat{g}$ is an element in $\cC^{\lambda}$.

\begin{lemma} \label{Qr-Qsleqx}
 $$ \sup_{Q\in\cD^\lambda}\left\{\langle Q^r,x+\widehat{g}\rangle - \langle Q^s,e_T\rangle \right\}= \langle \widehat{Q}^r,x+\widehat{g}\rangle - \langle\widehat{Q}^s,e_T\rangle = x.$$
\end{lemma}

\begin{proof}
 Given a $Q \in \cD^\lambda$ which is a convex set, and an $\varepsilon \in (0,1)$, define 
   $$ Q_{\varepsilon} := (1-\varepsilon) \widehat Q + \varepsilon Q \in \cD^\lambda. $$ 
 It follows $Q_{\varepsilon}^r = (1-\varepsilon) \widehat Q^r + \varepsilon Q^r $.
 By the opitimality of $\widehat Q$ and the convexity of $V(\cdot)$, we have
 
 \begin{equation*}
  \begin{aligned}
   0 &\geq \frac{1}{\varepsilon\widehat{y}}\left\{\massE\left[V\left(\widehat{y}\frac{d{\widehat Q}^r}{d\massP}\right)\right] + \widehat{y}\langle\widehat{Q},e_T\rangle 
                                       - \massE\left[V\left(\widehat{y}\frac{d{Q}^r_{\varepsilon}}{d\massP}\right)\right] - \widehat{y}\langle Q_{\varepsilon},e_T\rangle\right\} \\
     &=  \frac{1}{\varepsilon\widehat{y}}\massE \left[V\left(\widehat{y}\frac{d\widehat{Q}^r}{d\massP}\right) - V\left(\widehat{y}\frac{dQ^r_{\varepsilon}}{d\massP}\right)\right] 
                                         + \langle \widehat{Q},e_T\rangle - \langle {Q},e_T\rangle \\
     &\geq  \frac{1}{\varepsilon\widehat{y}}\massE \left[\widehat{y}\left(\frac{d\widehat{Q}^r}{d\massP}- \frac{dQ^r_{\varepsilon}}{d\massP}\right)V'\left(\widehat{y}\frac{dQ^r_{\varepsilon}}{d\massP}\right)\right] 
                                         + \langle \widehat{Q},e_T\rangle - \langle {Q},e_T\rangle \\   
     &= \massE\left[\left(\frac{dQ^r}{d\massP}-\frac{d{\widehat Q}^r}{d\massP}\right) I\left(\widehat{y}\frac{dQ^r_{\varepsilon}}{d\massP}\right)\right]
                + \langle \widehat{Q},e_T\rangle - \langle {Q},e_T\rangle. \\   
  \end{aligned}
 \end{equation*}
 
 We now claim that $\left(\left(\frac{dQ^r}{d\massP}-\frac{d{\widehat Q}^r}{d\massP}\right)I\left(\widehat{y}\frac{dQ^r_{\varepsilon}}{d\massP}\right) \right)^- $ is uniformly integrable. 
 Indeed,  
 \begin{equation*}
  \begin{aligned}
   \left(\left(\frac{dQ^r}{d\massP}-\frac{d{\widehat Q}^r}{d\massP}\right)I\left(\widehat{y}\frac{dQ^r_{\varepsilon}}{d\massP}\right) \right)^-
      \leq \frac{d{\widehat Q}^r}{d\massP}I\left( \widehat y \frac{d{ Q}^r_{\varepsilon}}{d\massP} \right)
      \leq \frac{d{\widehat Q}^r}{d\massP}I\left( \widehat y (1- \varepsilon) \frac{d{ Q}^r_{\varepsilon}}{d\massP} \right),    
  \end{aligned}
 \end{equation*}
 where the last term is uniformly integrable for sufficiently small $\varepsilon$ by Lemma \ref{UI}. 
 Hence we can apply Fatou's lemma, and obtain
 \begin{equation*}
  \begin{aligned}
   0 &\geq \liminf_{\varepsilon\to 0}\massE \left[\left(\frac{d{ Q}^r}{d\massP} - \frac{d{\widehat Q}^r}{d\massP}\right)I\left(\widehat{y}\frac{d{ Q}^r_{\varepsilon}}{d\massP}\right)\right]
                +\langle \widehat{Q},e_T\rangle -\langle {Q},e_T\rangle \\    
     &\geq \massE\left[\left(\frac{dQ^r}{d\massP} - \frac{d{\widehat Q}^r}{d\massP} \right) I \left(\widehat{y}\frac{d{ \widehat Q}^r}{d\massP} \right) \right]
                + \langle \widehat{Q},e_T\rangle - \langle {Q},e_T\rangle \\    
     &= \langle Q^r,x+\widehat{X}\rangle - \langle \widehat{Q}^r,x+\widehat{X}\rangle + \langle \widehat{Q}^s,e_T\rangle-\langle Q^s,e_T\rangle,  
  \end{aligned}
 \end{equation*}
 which implies our assertion. 
\end{proof}

\begin{lemma} 
 $\widehat{g}\in\cC^{\lambda}$.
\end{lemma}

\begin{proof}
  Firstly, we show that $\widehat{g}\wedge n\in\cC^{\lambda}$ for all $n\in\NN$. 
  
  Since $\widehat{g}$ is uniformly bounded from below, $\widehat{g}\wedge n\in L^{\infty}$. 
  For any $Q\in\cD^{\lambda,r}$, we have $Q^r=Q$. 
  It follows from Lemma \ref{Qr-Qsleqx} and $Q^s=0$ that 
  $$ \langle Q,x+\widehat{g}\wedge n\rangle \leq \langle Q,x+\widehat{g}\rangle \leq x+ \langle Q^s,e_T\rangle = x. $$
  Therefore
  $$ \langle Q,\widehat{g}\wedge n\rangle \leq x- \langle Q,x\rangle = 0, $$
  for all $Q\in\cD^{\lambda,r}$ and $n\in\NN$. 
  By Lemma \ref{XincC}, $\widehat{g}\wedge n\in\cC^{\lambda}$. 
  As, by \cite[Theorem 3.4]{Sch14super}, $\cC^{\lambda}_0$ is closed with respect to convergence in measure, and $\widehat{g}\wedge n\to\widehat{g}$ almost surely, 
    we have $\widehat{g}\in\cC^{\lambda}$. 
\end{proof}

\begin{proof}[Proof of main theorem]
 Since $\widehat{g}\in\cC^{\lambda}$ bounded from below, we have that $\langle \widehat{Q},\widehat{g}\rangle\leq 0$. 
 By \eqref{-x=-Qr,x+X+QseT} and the positivity of $x+\widehat{g}+e_T$, we get
 \begin{equation*}
  \begin{aligned}
   \langle \widehat{Q}, e_T\rangle + x &= \langle \widehat{Q}^r, x+\widehat{g}+e_T\rangle \leq \langle \widehat{Q}, x+\widehat{g}+e_T\rangle  \\
          &\leq \langle \widehat{Q}, e_T\rangle + \langle \widehat{Q}, x\rangle \leq \langle \widehat{Q}, e_T\rangle + x, 
  \end{aligned}
 \end{equation*}
 which implies 
  $$ \langle \widehat{Q}^s, x+\widehat{g}+e_T\rangle = 0, \quad \langle \widehat{Q}, \widehat{g}\rangle = 0, \quad \langle \widehat{Q}, x\rangle =x. $$
 Together with 
  $$ x+\widehat{g}+e_T = I\left(\widehat{y}\frac{d\widehat{Q}^r}{d\massP}\right) $$
  we get equalities instead of inequalities in \eqref{UleqJ+xy}, i.e.,
  $$ \massE[U(x+\widehat{g}+e_T)] = \massE\left[V\left(\widehat{y}\frac{d\widehat{Q}^r}{d\massP}\right)\right] 
                                    + \widehat{y}\langle \widehat{Q},e_T\rangle + x\widehat{y}.  $$
 Hence for $x>x_0$, we have 
 \begin{equation*}
  \begin{aligned}
   u(x) &\geq \massE[U(x+\widehat{g}+e_T)] = \massE\left[V\left(\widehat{y}\frac{d\widehat{Q}^r}{d\massP}\right)\right] 
                                              + \widehat{y}\langle \widehat{Q},e_T\rangle + x\widehat{y} \\
        &\geq v(\widehat{y})+x\widehat{y} = u(x), 
  \end{aligned}
 \end{equation*}
 which shows the optimality of $\widehat{g}\in\cC^{\lambda}$ and \eqref{u=v+xy}. 
 Since $u$ is differentiable, \eqref{v=u-xy} follows from the convex duality theory. 
 
 By the positivity of $x+\widehat{g}+e_T$, we obtain that
 $$ u(x)=\massE[U(x+\widehat{g}+e_T)]> -\infty, $$
 for all $x>x_0$, which implies the existence of an $g\in \mathcal C_0^{\lambda}$ such that $x+g+e_T>0$ almost surely, 
 hence $\langle Q, x+g+e_T \rangle \geq 0$, and therefore 
  \begin{eqnarray*}
   x\geq \langle Q, x \rangle 
    \geq \langle Q, x \rangle + \langle Q, g \rangle
    \geq \langle Q, -e_T\rangle,
  \end{eqnarray*}
  for all $Q\in \mathcal D^\lambda$, which follows that $$x_0 \geq \sup_{Q \in \mathcal D^\lambda} \langle Q,-e_T \rangle .$$
  By lemma \ref{inadav}, we have that
  $$x_0 = \sup_{Q \in \mathcal D^\lambda} \langle Q,-e_T \rangle,$$
   which completes the proof. 
\end{proof}

\bibliography{20131204}

\begin{thebibliography}{10}

\bibitem{BC12}
G.~Benedetti and L.~Campi.
\newblock {Multivariate utility maximization with proportional transaction
  costs and random endowment}.
\newblock {\em SIAM Journal on Control and Optimization}, 50(3):1283--1308,
  2012.

\bibitem{Bou02}
B.~Bouchard.
\newblock {Utility maximization on the real line under proportional transaction
  costs}.
\newblock {\em Finance and Stochastics}, 6(4):495--516, 2002.

\bibitem{CO11}
L.~Campi and M.P. Owen.
\newblock {Multivariate utility maximization with proportional transaction
  costs}.
\newblock {\em Finance and Stochastics}, 15(3):461--499, 2011.

\bibitem{CS06}
L.~Campi and W.~Schachermayer.
\newblock {A super-replication theorem in {K}abanov's model of transaction
  costs}.
\newblock {\em Finance and Stochastics}, 10(4):579--596, 2006.

\bibitem{CK96}
J.~Cvitani{\'c} and I.~Karatzas.
\newblock {Hedging and portfolio optimization under transaction costs: a
  martingale approach}.
\newblock {\em Mathematical Finance}, 6(2):133--165, 1996.

\bibitem{CSW01}
J.~Cvitani{\'c}, W.~Schachermayer, and H.~Wang.
\newblock {Utility maximization in incomplete markets with random endowment}.
\newblock {\em Finance and Stochastics}, 5:259--272, 2001.

\bibitem{CW01}
J.~Cvitani{\'c} and H.~Wang.
\newblock {On optimal terminal wealth under transaction costs}.
\newblock {\em Journal of Mathematical Economics}, 35(2):223--231, 2001.

\bibitem{CS15duality}
C.~Czichowsky and W.~Schachermayer.
\newblock {Duality theory for portfolio optimisation under transaction costs}.
\newblock {\em to appear in Annals of Applied Probability}, 2015.

\bibitem{CSY15}
C.~Czichowsky, W.~Schachermayer, and J.~Yang.
\newblock {Shadow prices for continuous processes}.
\newblock {\em to appear in Mathematical Finance}, 2015.

\bibitem{DPT01}
G.~Deelstra, H.~Pham, and N.~Touzi.
\newblock {Dual formulation of the utility maximization problem under
  transaction costs}.
\newblock {\em Annals of Applied Probability}, 11(4):1353--1383, 2001.

\bibitem{Kab99}
Yu.M. Kabanov.
\newblock {Hedging and liquidation under transaction costs in currency
  markets}.
\newblock {\em Finance and Stochastics}, 3(2):237--248, 1999.

\bibitem{KZ03}
I.~Karatzas and G.~\v{Z}itkovi{\'c}.
\newblock {Optimal consumption from investment and random endowments in
  incomplete semimartingale markets}.
\newblock {\em The Annals of Probability}, 31(4):1821--1858, 2003.

\bibitem{KS99}
D.~Kramkov and W.~Schachermayer.
\newblock {The asymptotic elasticity of utility functions and optimal
  investment in incomplete markets}.
\newblock {\em The Annals of Applied Probability}, 9:904--950, 1999.

\bibitem{LY16}
Y.~Lin and J.~Yang.
\newblock {Utility maximization problem with random endowment and transaction
  costs: when wealth becomes negative}.
\newblock {\em Preprint}, 2016.

\bibitem{OZ09}
M.P. Owen and G.~\v{Z}itkovi{\'c}.
\newblock {Optimal investment with an unbounded random endowment and
  utility-based pricing}.
\newblock {\em Mathematical Finance}, 19(1):129--159, 2009.

\bibitem{Sch14admissible}
W.~Schachermayer.
\newblock {Admissible trading strategies under transaction costs}.
\newblock In {\em {S{\'e}minaire de Probabilit{\'e}s XLVI}}, volume 2123 of
  {\em {Lecture Notes in Mathematics}}, pages 317--331. Springer International
  Publishing Switzerland, 2014.

\bibitem{Sch14super}
W.~Schachermayer.
\newblock {The super-replication theorem under proportional transaction costs
  revisited}.
\newblock {\em Mathematics and Financial Economics}, 8(4):383--398, 2014.

\end{thebibliography}

\bibliographystyle{plain}

\end{document}